\documentclass[pra,twocolumn, floatfix,aps,nofootinbib]{revtex4-1}
\usepackage{amsmath,verbatim,enumerate}
\usepackage{bbm,comment}
\usepackage{amssymb}
\usepackage{graphicx}
\usepackage{times}
\usepackage{color}
\usepackage[amsmath,hyperref,standard,thmmarks]{ntheorem}
\usepackage{times}
\usepackage[normalem]{ulem}
\usepackage{etex}
\renewtheorem{theorem}{Theorem}
\renewtheorem{proposition}{Proposition}

\renewtheorem{lemma}[theorem]{Lemma}
\renewtheorem{corollary}[theorem]{Corollary}
\renewtheorem{definition}[theorem]{Definition}

\renewtheorem{example}[theorem]{Example}
\renewtheorem{remark}[theorem]{Remark}

\newtheorem{result}[theorem]{Result}

\definecolor{henrik}{rgb}{1,.0,.4}

\newcommand{\RR}{\mathbb{R}}

\newcommand{\id}{\mathbbm{1}}
\newcommand{\mc}[1]{\mathcal{#1}}

\newcommand{\ketbra}[2]{| #1 \rangle \langle #2 |}

\newcommand{\proj}[1]{\vert #1\rangle\!\langle#1 \vert}

\newcommand{\norm}[1]{\left\Vert #1 \right\Vert}
\def\one{\mathbbm{1}}

\newcommand{\Tr}{\operatorname{tr}}
\newcommand{\tr}{\Tr}

\newcommand{\fu}{Dahlem Center for Complex Quantum Systems, Freie Universit{\"a}t Berlin, 14195 Berlin, Germany}

\begin{document}
\title{Axiomatic characterization of the quantum relative entropy and free energy}

\author{Henrik Wilming, Rodrigo Gallego, Jens Eisert}
\affiliation{\fu}

\begin{abstract}
Building upon work by Matsumoto, we show that the quantum relative entropy with full-rank second argument is determined by four simple axioms:
i) Continuity in the first argument, ii) the validity of the data-processing inequality, iii) additivity under tensor products, and iv) super-additivity.
This observation has immediate
implications for quantum thermodynamics, which we discuss. Specifically,
we demonstrate that, under reasonable restrictions, the free energy is singled out as a
measure of athermality. In particular, we consider an extended class of Gibbs-preserving maps as free operations in a resource-theoretic
framework, in which a catalyst is allowed to build up correlations with the system at hand. The free energy is the only extensive and
continuous function that is monotonic under such free operations.
\end{abstract}

\maketitle
\section{Introduction}
The quantum relative entropy captures the statistical distinguishability of two quantum states. For two
states $\rho$ and $\sigma$ supported on the same Hilbert space it is defined as
\begin{equation}\label{eq:relent}
S(\rho || \sigma) = \tr\left(\rho \log \rho - \rho \log \sigma \right),
\end{equation}
whenever $\mathrm{supp}(\rho) \subseteq \mathrm{supp}(\sigma)$ and set to infinity otherwise. This quantity has a clear
interpretation in the statistical discrimination of $\rho$ from $\sigma$,
appearing as an error rate in quantum hypothesis testing \cite{Hiai2005,Ogawa2005,GeneralStein},
a result commonly known as
Stein's Lemma. It is hence no surprise that this quantity appears in a plethora of places in contemporary quantum physics. This is particularly
 true in the context of quantum information theory \cite{Vedral2002}.
In the  relative entropy of entanglement it quantifies the entanglement content of a general quantum state \cite{Vedral1997}.
More generally, it appears in conversion rates in so-called resource theories \cite{Brandao2013,Brandao2015}.
Relatedly, it takes center stage in the
 problem of (approximately) recovering quantum information \cite{Junge2016}. But the applications are not confined
 to quantum information theory. In many-body physics, it provides bounds on the clustering of correlations in space in terms of the mutual information \cite{Kastoryano2013,Bernigau2015}. In quantum thermodynamics \cite{Goold2016}, which is the context that is in the focus
 of attention in this note,
 its interpretation as the non-equilibrium free energy gives an upper bound to how much work can be extracted from a non-equilibrium system and is important in answering how to operationally define work in the quantum regime in the first place \cite{Gallego2015}. Not the least,
 it has appeared in the context of the AdS/cft correspondence \cite{Lashkari2015}, again drawing from and building upon the
 above mentioned applications.

In this note, we restrict to the case where the second argument $\sigma$ has full rank and only consider finite dimensional Hilbert spaces. Essentially by re-interpreting and building upon a theorem by Matsumoto \cite{Matsumoto2010}, we will show that the quantum relative entropy \eqref{eq:relent} is (up to a constant factor) the only function featuring the following four properties:
\begin{enumerate}
\item \label{prop:continuity} \emph{Continuity}: For fixed $\sigma$, the map $\rho\mapsto S(\rho || \sigma)$ is
continuous \cite{AudenaertContinuous}.
\item \label{prop:dpi} \emph{Data-processing inequality}: For any quantum channel $T$ we have,
\begin{equation}
	S(T(\rho)|| T(\sigma)) \leq S(\rho || \sigma).
	\end{equation}
\item \label{prop:additivity} \emph{Additivity}:
\begin{equation}
S(\rho_1\otimes \rho_2 || \sigma_1\otimes \sigma_2) = S(\rho_1||\sigma_1) + S(\rho_2||\sigma_2).
	\end{equation}
\item \label{prop:super-add} \emph{Super-additivity}: For any bipartite state $\rho_{1,2}$ with marginals $\rho_1,\rho_2$ we have
\begin{equation}
S(\rho_{1,2} || \sigma_1\otimes \sigma_2) \geq S(\rho_1||\sigma_1) + S(\rho_2||\sigma_2).
	\end{equation}
\end{enumerate}
No subset of these properties characterizes the relative entropy uniquely, but the Properties \ref{prop:continuity} to \ref{prop:additivity} are, for example, also fulfilled by the Renyi-divergences \cite{Tomamichel2016}.

The uniqueness of the quantum relative entropy under Properties \ref{prop:continuity}-\ref{prop:super-add} has significant
implications for \emph{quantum thermodynamics} (QT), which we elaborate upon. The formalism of QT has recently
been recast in within the framework of a resource theory \cite{Janzing00,Brandao2013,Horodecki2013}, so one in which
quantum states that are different from Gibbs states (at a fixed environment temperature) are considered resources. We will refer here to this kind of resource as
\emph{athermality}. Within this resource theory one is, among other problems, interested in finding \emph{bona fide} measures of athermality.
These are functions that quantify the amount of athermality of a given system.
A requirement for a function reasonably quantifying the degree of athermality is that the
it does not increase under the free operations of the resource theory. The problem of identifying such functions
has been studied intensively in the last years for different classes of free operations, providing families of valid measures that are regarded as generalizations of the \emph{free energy} \cite{Brandao15,Brandao2015,Gour16,Buscemi16}. They share the property that they are
all based of generalizations of the quantum relative entropy \eqref{eq:relent}.

In this work, we will use the uniqueness result on the quantum relative entropy to show that the usual non-equilibrium free energy emerges as the unique \emph{continuous} and \emph{extensive} measure of athermality under a certain meaningful choice of free operations.
In this sense, we also provide a fresh link of resource-theoretic considerations in quantum
thermodynamics to more traditional descriptions of thermodynamic processes in the quantum regime.

\section{Axiomatic derivation of quantum relative entropy}

We start by formally stating the main technical result.

\begin{theorem}[Uniqueness theorem]\label{thm:maintheorem}
Let $f$ be a function on pairs of quantum states acting on the same finite dimensional Hilbert space, with the second argument having full rank. Suppose $f$ fulfills Properties \ref{prop:continuity}-\ref{prop:super-add}. Then it is given by
\begin{equation}
f(\rho,\sigma)= C \tr\left(\rho \log \rho - \rho \log \sigma \right) := C S(\rho \| \sigma),
\end{equation}
for some constant $C>0$.
\end{theorem}

The proof relies on a characterization of the relative entropy in terms of different properties laid out in Ref.\
\cite{Matsumoto2010}. To state it, we first require a definition: Let
$(\rho,\sigma)$ be a pair of states on a finite-dimensional Hilbert space $\mc{H}$ and $\{\rho'_n\}$ be a sequence of states on the
Hilbert spaces $\mc{H}^{\otimes n}$. We define a function $f$ on pairs of quantum states to be
\emph{lower asymptotically semi-continuous} (l.a.s.) with respect to
$\sigma$ if
\begin{equation}
    \lim_{n\rightarrow \infty}\norm{\rho^{\otimes n}-\rho'_n}_1 =0
\end{equation}
implies
\begin{equation}
    \liminf_{n\rightarrow \infty}
    \frac{1}{n}(f(\rho'_n,\sigma^{\otimes n})-f(\rho^{\otimes n},\sigma^{\otimes n}))\geq 0.
\end{equation}
Then Matsumoto's theorem \cite{Matsumoto2010} for the relative entropy can stated in the following way.
\begin{theorem}[Matsumoto]\label{thm:matsumoto}
Let $f$ fulfill the data-processing inequality, additivity and be lower asymptotically semi-continuous with respect to all $\sigma$.
Then $f\propto S$.
\end{theorem}
The proof of Theorem \ref{thm:maintheorem} follows from the subsequent Lemma, which in turn implies that the Properties \ref{prop:continuity}-\ref{prop:super-add} give rise to the conditions of Theorem \ref{thm:matsumoto}.
\begin{lemma}[Lower asymptotically semi-continuity]
Let $f$ be a function on pairs of quantum states with the following properties,
\begin{itemize}
    \item The map $\rho\mapsto f(\rho,\sigma)$ is continuous for any fixed
        $\sigma$.
    \item Additivity: $f(\rho_1\otimes \rho_2,\sigma_1\otimes\sigma_2)=\sum_{i=1}^2
        f(\rho_i,\sigma_i)$.
    \item Super-additivity:
    \begin{equation}
    	f(\rho_{1,2},\sigma_1\otimes\sigma_2)\geq f(\rho_1\otimes \rho_2,\sigma_1\otimes\sigma_2).
\end{equation}
\end{itemize}
Then $f$ is lower asymptotically semi-continuous with respect to any
$\sigma$.
\end{lemma}
\begin{proof}
    Let $\{\rho'_n\}$ be a sequence of states such that
    $\norm{\rho'_n-\rho^{\otimes n}}_1\rightarrow 0$. Since the trace norm
    fulfills the data-processing inequality, we know that
    $||\rho'_{n,i}-\rho||_1 \rightarrow 0$, where $\rho'_{n,i}$ denotes
    the marginal of $\rho'_n$ on the $i$-th tensor-factor. Hence, the
    marginals converge to $\rho$.
    From the properties of $f$, we furthermore see that
    \begin{align}
        &\frac{1}{n}\left(f(\rho'_n,\sigma^{\otimes n})-f(\rho^{\otimes n},\sigma^{\otimes n}) \right)\\
        &\geq \frac{1}{n}\sum_i\left(f(\rho'_{n,i},\sigma)-f(\rho,\sigma)\right)\\
        &\geq \min_i\{f(\rho'_{n,i},\sigma)\} - f(\rho,\sigma) \overset{n\rightarrow\infty}{\longrightarrow} 0,
    \end{align}
    where the limit follows from continuity and the second line from additivity and super-additivity.
\end{proof}

\section{Uniqueness of the free energy}

The results of the previous section, in particular Theorem \ref{thm:maintheorem}, can be applied to the resource theory of 
$\beta$-athermality. We formulate it as a resource theory of pairs of a quantum state and a Hamiltonian $(\rho,H)$ that 
we call \emph{object}. An object $(\rho,H)$ is said to have the resource of $\beta$-athermality if it fulfills
\begin{equation}
\rho\neq \omega_{\beta,H}
\end{equation}
where $\omega_{\beta,H}$ is the Gibbs state for the Hamiltonian $H$ and inverse temperature $\beta>0$, given by
\begin{equation}\label{eq:Gibbs}
\omega_{\beta,H}:=\frac{e^{-\beta H}}{\tr (e^{-\beta H})}.
\end{equation}
In this way, the resource theory of $\beta$-atermality is concerned with the absence of thermal equilibrium at temperature $1/\beta$ \cite{Brandao2013}. Concerning the set of free operations, we will be considering here the most general set of operations that do not create resourceful states from states featuring fewer resources.

In order to progress, let us first define the so-called \emph{Gibbs-preserving maps (GP)} which are quantum channels that have Gibbs state \eqref{eq:Gibbs} as a fixed point. More formally, a GP-channel is defined as a trace-preserving completely positive map $\mc{G}_{\beta}$
with the property that
\begin{equation}\label{eq:gpchannels}
\mc{G}_{\beta}(\omega_{\beta,H})= \omega_{\beta,H} \;\; \forall \: H.
\end{equation}
Note that formulated as above, GP channels only induce transitions that change the quantum state but not the Hamiltonian. This can be extended by simply considering functions $G$ that act on the object, possibly changing also the Hamiltonian, but which at the same time
 do not create $\beta$-athermality. In this way, we define a GP-map as a function $(\rho,H) \mapsto (\sigma,K)=G_{\beta}(\rho,H)$ such that
\begin{equation}\label{eq:gpmaps}
G_{\beta}(\omega_{\beta,H},H)=(\omega_{\beta,K},K).
\end{equation}
This condition can equivalently be cast into the following form: One may define the set of GP channels
as $\mc{G}_{\beta}^{H} (\omega_{\beta,H})=\omega_{\beta,K(H)}$ for all $H$, and the map between Hamiltonians as
$\bar{\mc{G}}(H)=K$ so that
\begin{equation}\label{eq:gpmapsentries}
G(\rho,H)=(\mc{G}_{\beta}^{H}(\rho),\bar{\mc{G}}(H)).
\end{equation}
With this notation, condition \eqref{eq:gpmaps} is simply given by
\begin{equation}\label{eq:conditiongpmaps}
\mc{G}_{\beta}^{H}(\omega_{\beta,H})=\omega_{\beta,\bar{\mc{G}}(H)}.
\end{equation}
GP-maps $G$ are not only a natural extension of GP-channels $\mc{G}$ for the case where Hamiltonians are modified, but one can also
see that any GP-map can be implemented if one is given access to a GP-channels and an ancillary system in a Gibbs state. This is formalized by the following Lemma taken from Ref.\ \cite{Gallego2015}.

\begin{lemma}[Implementation of GP maps \cite{Gallego2015}]
Any map $G_{\beta}$ fulfilling \eqref{eq:gpmaps} acting on a system $S$  can be implemented by adding an ancillary system $A$ in the Gibbs state $(\omega_{\beta,K},K)$ and applying a GP channel $\mc{G}$ to the entire compound. More formally, we find that
\begin{equation}
G_{\beta}(\rho_S,H_S):=(\sigma,K)=\big(\tr_S (\mc{G}_{\beta}(\rho \otimes \omega_{\beta,K})),K\big).
\end{equation}
\end{lemma}
Once we have established the set of GP maps for the objects,
we will now introduce the notion of catalyst in this framework. This is done analogously to the case of catalysts for other sets of operations such as \emph{thermal operations} \cite{Horodecki2013,Brandao2015,Lostaglio15,Ng15}. In the following, we will also frequently drop the $\beta$-subscript from GP-maps for simplicity of notation.

\subsection{Catalysts and correlations}
We will now turn to defining the transitions between objects that can be performed with GP maps and the use of what is called a ``catalyst''
in this context. This is simply an ancillary quantum system that is left in the same state (in a sense that will be made precise later) after the transition is performed, rendering the metaphor of an actual catalyst quite appropriate.

\begin{definition}[Catalytic free transition]\label{def:catalyticfreetransition}
We say that the transition
\begin{equation}
(\rho_S,H_S) \rightarrow (\sigma_S,K_S)
\end{equation}
is a \emph{catalytic free transition} if there exist a GP map $G$ and a system $A$ described by the object $(\gamma_A,R_A)$ such that
\begin{equation}
G\big( (\rho_S,H_S) \otimes (\gamma_A,R_A) \big)= (\sigma_S,K_S) \otimes (\gamma_A,R_A).
\end{equation}
We will in this case simply denote it by
\begin{equation}
(\rho_S,H_S)\stackrel{\text{c}}{>} (\sigma_S,K_S).
\end{equation}
\end{definition}
Here, we are employing the convenient notation
\begin{equation}
 (\rho_S,H_S) \otimes (\gamma_A,R_A) :=(\rho_S \otimes \gamma_A , H_S \otimes \id_A+ \id_S \otimes R_A)
 \end{equation}
 to describe tensor products of objects. In the remainder of this work, we will simply write $H_S \otimes \id_A+ \id_S \otimes R_A:=H_S+R_A$. Importantly, we are assuming that the catalyst $A$ is left in the same state and Hamiltonian and also uncorrelated from $S$. The role of the correlations of the catalysts and its role in quantum thermodynamics has been first noted in Ref.\ \cite{Lostaglio15}. There, one considers a catalysts consisting of $k$ subsystems and one merely
 demands that the marginal state of each subsystem is left untouched. We define it here formally for the case of GP maps.
\begin{definition}[Marginal-catalytic free transition \cite{Lostaglio15}]\label{def:marginal-catalyticfreetransition}
We say that the transition
\begin{equation}
(\rho_S,H_S) \rightarrow (\sigma_S,K_S)
\end{equation}
is a \emph{marginal-catalytic free transition} if there exist a GP map $G$ and systems $A_1,\ldots,A_k$ described by the object $(\gamma_A,R_A)=\bigotimes_{i=1}^k (\gamma^i,R^i)$ such that
\begin{equation}
G\big( (\rho_S,H_S) \otimes (\gamma_A,R_A) \big)= (\sigma_S,K_S) \otimes (\tilde{\gamma}_A,R_A),
\end{equation}
where $\tr_{|A_i}(\tilde{\gamma}_A)=\tr_{|A_i}(\gamma_A)$ for all $i \in (1,\ldots,k)$. We will in this case simply denote it by
\begin{equation}
(\rho_S,H_S)\stackrel{\text{mc}}{>} (\sigma_S,K_S).
\end{equation}
\end{definition}
Note that in this case the system $A$ does not remain unchanged, but only its local marginals. In this sense, it is not truly a catalyst,
but a catalyst on its reduced states. It is natural to expect that this indeed allows for a larger set of transitions, since the system $A$ is ``used up'' by employing the initial lack of correlations as a resource.

We will now consider a family of transitions that also introduces correlations, but for which the catalyst is, unlike in Definition \ref{def:marginal-catalyticfreetransition}, left entirely untouched. In this case, correlations are built up between the system and
the catalyst. In this way, the catalyst is re-usable as long as it is employed in order to implement a transition on a new system. We call this transitions, originally introduced in Ref.\ \cite{Gallego2015}, \emph{correlated-catalytic free transitions}:

\begin{definition}[Correlated-catalytic free transition]\label{def:correlated-catalyticfreetransition}
We say that the transition
\begin{equation}
(\rho_S,H_S) \rightarrow (\sigma_S,K_S)
\end{equation}
is a \emph{correlated-catalytic free transition} if there exist a GP map $G$ and a system $A$ described by the object $(\gamma_A,R_A)$ such that
\begin{equation}
G\big( (\rho_S,H_S) \otimes (\gamma_A,R_A) \big)= (\eta,K_S+ R_A),
\end{equation}
where $\tr_{A}(\eta)=\sigma_S$ and $\tr_{S}(\eta)=\gamma_A$. We will in this case simply denote it by
\begin{equation}
(\rho_S,H_S)\stackrel{\text{cc}}{>} (\sigma_S,K_S).
\end{equation}
\end{definition}
 We will now show that the non-equilibrium free energy is the only function, under reasonable assumptions, that does not increase under operations of the form of Definitions \ref{def:marginal-catalyticfreetransition} and \ref{def:correlated-catalyticfreetransition}.

\subsection{Free energy as a unique measure of non-equilibrium}
We will call a \emph{measure of non-equilibrium} a function that quantifies how far a given object $(\rho,H)$ is from its equilibrium object $(\omega_{\beta,H},H)$. The minimal requirement on such a measure is that it is non-increasing under free transitions. The larger the set of free transitions, the more restricted is the allowed set of measures. One of the most well-studied measures of non-equilibrium is based on the quantum relative entropy. It is related to the free energy as
\begin{equation}
\Delta F_{\beta}(\rho,H):= \frac{1}{\beta}S(\rho\|\omega_{\beta,H}) =F_{\beta}(\rho,H)- F_{\beta}(\omega_{\beta,H},H)
\end{equation}
where $F_{\beta}(\rho,H)=\tr(\rho H) -\beta^{-1} S(\rho)$ with $S$ being the von Neumann (and not the relative) entropy. The measure $\Delta F_{\beta}$ fulfills  the following properties that we express here for a generic measure denoted by $M_{\beta}$:
\begin{enumerate}[I]
\item \label{prop:fe:continuity} \emph{Continuity}: For fixed Hamiltonian $H$, the map $\rho\mapsto M_{\beta}(\rho,H)$ is continuous.
\item \label{prop:fe:additivity} \emph{Additivity}:
\begin{equation}
M_{\beta}(\rho_1\otimes\rho_2, H_1+ H_2) = M_{\beta}(\rho_1,H_1) +M_{\beta}(\rho_2,H_2).\nonumber
\end{equation}
\item \label{prop:fe:dpi} \emph{Monotonicity}:
\begin{enumerate}
\item \emph{Monotonicity: }\\ \label{prop:monotonicity}
$M_{\beta}(\rho,H)\geq M_{\beta}(\sigma,K)$ if $G(\rho,H) = (\sigma,K)$.
\item \emph{Catalytic monotonicity: }\\ \label{prop:catalyticmonotonicity}
$M_{\beta}(\rho,H)\geq M_{\beta}(\sigma,K)$ if $(\rho,H) \stackrel{\text{c}}{>} (\sigma,K)$.
\item \emph{Marginal-catalytic monotonicity: }\\ \label{prop:mcatalyticmonotonicity}
$M_{\beta}(\rho,H)\geq M_{\beta}(\sigma,K)$ if $(\rho,H) \stackrel{\text{mc}}{>} (\sigma,K)$.
\item \emph{Correlated-catalytic monotonicity: }\\ \label{prop:ccatalyticmonotonicity}
$M_{\beta}(\rho,H)\geq M_{\beta}(\sigma,K)$ if $(\rho,H) \stackrel{\text{cc}}{>} (\sigma,K)$.
\end{enumerate}
\end{enumerate}
All those properties apply for all states and Hamiltonians involved. The fact that $\Delta F_{\beta}$ fulfills \ref{prop:fe:continuity} and \ref{prop:fe:additivity} follows from the continuity and additivity properties
of the quantum relative entropy. The other properties can be related to the data processing inequality and super-additivity as we will see in Theorem \ref{thm:equivalencemaintext}. Before that, let us note that for any function $M_{\beta}$ on objects, we can define a function $\mc{M}_{\beta}$ on pairs of quantum states as $\mc{M}_{\beta}(\rho,\omega_{\beta,H})=M_\beta(\rho,H)$. At the same time, it is true that
any full-rank state $\sigma$ can be thought of as the Gibbs-state of the modular Hamiltonian
\begin{align}\label{eq:modular_hamiltonian}
H_\sigma := -\frac{1}{\beta} \log \sigma+ C,
\end{align}
{for any $C\in \RR$}. With this notation, all objects of the form $(\sigma,H_\sigma)$ are Gibbs-objects. {Importantly, the modular Hamiltonian $H_\sigma$ is only defined up to an additive constant. It turns out however, that the properties \ref{prop:fe:additivity} and \ref{prop:fe:dpi} imply that $M_\beta(\rho,H) = M_\beta(\rho,H+C\one)$ for any $C\in \RR$ (see Appendix \ref{sec:appgauge} for a proof). Any additive measure of athermality is hence automatically gauge-invariant in this sense.}

Thus, the functions $\mc{M}_{\beta}$ and $M_{\beta}$ are in a one-to-one correspondence. With this equivalence, we say that a measure $\mc{M}_{\beta}$ is super-additive if, for any bipartite quantum states,
it fulfills
\begin{equation}\label{eq:superadditiveM}
\mc{M}_{\beta}(\rho_{1,2},\sigma_1\otimes \sigma_2)\geq \mc{M}_{\beta}(\tr_2(\rho)\otimes \tr_1(\rho),\sigma_1\otimes \sigma_2)
\end{equation}
and additive if it fulfills
\begin{equation}\label{eq:additiveM}
\mc{M}_{\beta}(\rho_1\otimes \rho_2,\sigma_1\otimes \sigma_2)= \mc{M}_{\beta}(\rho_1,\sigma_1) + \mc{M}_{\beta}(\rho_2,\sigma_2).
\end{equation}
Also, $\mc{M}_{\beta}$ is said to fulfill the data processing inequality if
\begin{equation}\label{eq:dataprocessingM}
\mc{M}_{\beta}(T(\rho),T(\sigma))\leq \mc{M}_{\beta}(\rho,\sigma)
\end{equation}
for all  $\rho$, all $\sigma$ being full-rank and for all quantum channels $T$.

At this point, a note of caution is appropriate. We have previously defined the function $\mc M_\beta$ only in the specific case where the second argument has full rank.
There clearly are quantum-channels $T$ that reduce the rank of full-rank states,
in which case $\mc M_\beta(T(\rho),T(\sigma))$ may at first seem undefined. This is not a problem, however. To see this, we make use of the following fact about quantum channels:
\begin{lemma}[Rank-decreasing quantum channels]\label{lemma:rank}
Let $T:\mc{B}(\mc H)\rightarrow \mc{B}(\mc H')$ be a quantum channel and $\sigma$ any full-rank state. If $T(\sigma)$ is only supported on a subspace $P\subseteq \mc H'$, then $T(\rho)$ is supported only within $P$ for any $\rho$.
\begin{proof}
The proof is given in the appendix.
\end{proof}
\end{lemma}
By the previous lemma, we see that any quantum channel that maps a full-rank state $\sigma$ into a state $T(\sigma)$ without full rank simply maps all states to the smaller Hilbert space $P=\mathrm{supp}(T(\sigma))$ and should be considered as a map from states on $\mc H$ to states on $P$ instead. Since the function $\mc{M}_\beta$ is defined on all finite-dimensional Hilbert spaces, we can simply assume that it acts on $\mc{B}(P)\times \mc{B}(P)$ in this case. In yet other words, the function $\mc M_\beta(\rho,\sigma)$ is always defined if $\mathrm{supp}(\rho)\subseteq \mathrm{supp}(\sigma)$ by restricting it to $\mathrm{supp}(\sigma)$,
\begin{equation}
\mc M_\beta(\rho,\sigma) = \mc M_\beta(\rho|_{\mathrm{supp}(\sigma)},\sigma|_{\mathrm{supp}(\sigma)}).
\end{equation}
We can then show the following equivalence between properties of $M_\beta$ and $\mc M_\beta$.

\begin{theorem}[Equivalence of $M_\beta$ and $\mc M_\beta$]\label{thm:equivalencemaintext} There exist the following two equivalences between the properties of the measures of athermality $M_{\beta}(\rho,H)$ and the corresponding function $\mc{M}_{\beta}(\rho,\sigma)$.
\begin{itemize}
\item The measure $M_{\beta}$ fulfills additivity \ref{prop:fe:additivity} and marginal-catalytic monotonicity \ref{prop:mcatalyticmonotonicity} $\iff$ $\mc{M}_\beta$ is super-additive \ref{eq:superadditiveM}, additive \ref{eq:additiveM} and fulfills the data-processing inequality \ref{eq:dataprocessingM}.
\item The measure $M_{\beta}$ fulfills Additivity \ref{prop:additivity} and correlated-catalytic monotonicity \ref{prop:ccatalyticmonotonicity} $\iff$ $\mc{M}_\beta$ is super-additive \ref{eq:superadditiveM}, additive \ref{eq:additiveM} and fulfills the data-processing inequality \ref{eq:dataprocessingM}.
\end{itemize}
\end{theorem}
The proof of this theorem, together with a more detailed set of the implications between the properties of $M_{\beta}$ and the corresponding function $\mc{M}_{\beta}$, is provided in the appendix.

The previous theorem simply tells us that any additive measure of athermality $M_{\beta}$ that does not increase under marginal-catalytic operations (Definition \ref{def:marginal-catalyticfreetransition}) or also under correlated-catalytic operations (Definition \ref{def:correlated-catalyticfreetransition}) is in one to one correspondence with a function $\mc{M}_{\beta}$ that is additive, super-additive and fulfills the data processing inequality. This has as a first consequence that the measure $\Delta F_{\beta}$ fulfills indeed Properties \ref{prop:fe:continuity}-\ref{prop:fe:dpi}. More importantly, using our re-formulation of Matsumoto's result of Theorem \ref{thm:maintheorem}, we can show that $\Delta F_{\beta}$ is, up to a constant factor, the \emph{only} measure of athermality that fulfills \ref{prop:fe:continuity}-\ref{prop:fe:dpi}. This is the content of our main result, which follows from Theorem \ref{thm:equivalencemaintext}.

\begin{result}[Uniqueness of monotones]
Any monotone for marginal-catalytic transitions or correlated-catalytic transitions at environment temperature $\beta$ that is additive and depends continuously on the density matrix is proportional to $\Delta F_\beta$.
\end{result}

The implications of this result are that the free energy difference $\Delta F_\beta$ is the only \emph{bona fide} quantifier of athermality under the most general set of free operations that do not create the resource.

\section{Discussion and Outlook}
In this work, we have investigated the question which properties uniquely determine the quantum relative entropy among all function on pairs of quantum states. Our re-formulation  of Matsumoto's result highlights the role of super-additivity as a key property in the axiomatic derivation of the quantum relative entropy. The role of super-additivity in the arena of quantum thermodynamics has been
shown to be related to the build up of correlations between the system at hand and a catalyst, which in turn represents the components of the machine that come back to their initial state after the cyclic process.

We have shown how the relative entropy and non-equilibrium free energy uniquely emerge from considerations about how to treat catalysts and their correlations in the resource theoretic approach to quantum thermodynamics. Usually, notions of relative entropy are employed to capture asymptotic weakly correlated settings (thermodynamic limit), thus when acting on many uncorrelated copies of a system (see Ref.\ \cite{Sparaciari2016} for a recent discussion of asymptotic thermodynamics from the point of view of resource theories).
Importantly, and in contrast, in our approach they emerge without having to invoke
any thermodynamic limit, but rather follow from properties of monotones in the single-shot setting. However, they are precisely singled out by the fact that we disregard correlations in the setting of marginal-catalytic and correlated-catalytic free transitions. It thus seems that the crucial feature of for the emergence of free energy is thus the disregarding of correlations. Note that this fits well to how these quantities appear in notions of macroscopic thermodynamics: Macroscopic equilibrium thermodynamics usually emerges in large systems which are well within thermodynamic phases. In such phases, correlations decay exponentially in space. Hence, the correlations of an object with its surrounding scale like its surface-area and not like their volume. Macroscopic objects are then essentially uncorrelated with other objects due to their small surface-to-volume ratio.

In this work,  we have distinguished two ways of creating correlations with the catalyst, the marginal-catalytic one of Definition \ref{def:marginal-catalyticfreetransition} and the correlated-catalytic one of Definition \ref{def:correlated-catalyticfreetransition}. The first represents the situation where the components of the machine become correlated among themselves, while the second represents the case where the machine builds correlation with the system upon which the machine induces a transition. Although these two sets both give rise to the free energy difference as a unique measure of athermality, we consider the latter as a much more adequate set of operations to incorporate correlations in thermodynamics. The reason for this is that the correlations build up between the catalyst and the system do not prevent one from re-using the catalyst to implement again a transition of the same kind on another system. It is an additional contribution of this
work to flesh out this difference.

We end the discussion by posing an interesting open question. This is to investigate how to characterize all the possible thermodynamic transitions that can be implemented with correlated-catalysts. In Ref.\ \cite{Gallego2015} we have seen that indeed the operations of Definition \ref{def:correlated-catalyticfreetransition} are more powerful than the ones of \ref{def:catalyticfreetransition}. At the same time, it has been recently been shown in Ref.\ \cite{Sparaciari2017} that a variant of Definition \ref{def:correlated-catalyticfreetransition} allows to extract work from passive states. The question remains whether all the transitions that do not increase the free energy difference $\Delta F_{\beta}$ are possible, as they indeed are for the ones of Definition \ref{def:marginal-catalyticfreetransition}, as shown in Ref.\ \cite{Lostaglio15}. If this is indeed true also for correlated-catalysts, one would have found an interpretation of the free energy as a unique criterium for the second law of thermodynamics. If it is not true, then it is necessary to consider genuinely new monotones, which are not additive or not continuous. Both options would be interesting from the perspective of the further development of quantum thermodynamics. 

{\bf{Acknowledgments}}: We acknowledge funding from the DFG (GA 2184/2-1), the BMBF, the
EU (COST, AQuS), and the ERC (TAQ) and the Studienstiftung des Deutschen Volkes.


\begin{thebibliography}{25}%
\makeatletter
\providecommand \@ifxundefined [1]{%
 \@ifx{#1\undefined}
}%
\providecommand \@ifnum [1]{%
 \ifnum #1\expandafter \@firstoftwo
 \else \expandafter \@secondoftwo
 \fi
}%
\providecommand \@ifx [1]{%
 \ifx #1\expandafter \@firstoftwo
 \else \expandafter \@secondoftwo
 \fi
}%
\providecommand \natexlab [1]{#1}%
\providecommand \enquote  [1]{``#1''}%
\providecommand \bibnamefont  [1]{#1}%
\providecommand \bibfnamefont [1]{#1}%
\providecommand \citenamefont [1]{#1}%
\providecommand \href@noop [0]{\@secondoftwo}%
\providecommand \href [0]{\begingroup \@sanitize@url \@href}%
\providecommand \@href[1]{\@@startlink{#1}\@@href}%
\providecommand \@@href[1]{\endgroup#1\@@endlink}%
\providecommand \@sanitize@url [0]{\catcode `\\12\catcode `\$12\catcode
  `\&12\catcode `\#12\catcode `\^12\catcode `\_12\catcode `\%12\relax}%
\providecommand \@@startlink[1]{}%
\providecommand \@@endlink[0]{}%
\providecommand \url  [0]{\begingroup\@sanitize@url \@url }%
\providecommand \@url [1]{\endgroup\@href {#1}{\urlprefix }}%
\providecommand \urlprefix  [0]{URL }%
\providecommand \Eprint [0]{\href }%
\providecommand \doibase [0]{http://dx.doi.org/}%
\providecommand \selectlanguage [0]{\@gobble}%
\providecommand \bibinfo  [0]{\@secondoftwo}%
\providecommand \bibfield  [0]{\@secondoftwo}%
\providecommand \translation [1]{[#1]}%
\providecommand \BibitemOpen [0]{}%
\providecommand \bibitemStop [0]{}%
\providecommand \bibitemNoStop [0]{.\EOS\space}%
\providecommand \EOS [0]{\spacefactor3000\relax}%
\providecommand \BibitemShut  [1]{\csname bibitem#1\endcsname}%
\let\auto@bib@innerbib\@empty
\bibitem [{\citenamefont {Hiai}\ and\ \citenamefont {Petz}(2005)}]{Hiai2005}%
  \BibitemOpen
  \bibfield  {author} {\bibinfo {author} {\bibfnamefont {F.}~\bibnamefont
  {Hiai}}\ and\ \bibinfo {author} {\bibfnamefont {D.}~\bibnamefont {Petz}},\
  }in\ \href {\doibase 10.1142/9789812563071_0004} {\emph {\bibinfo {booktitle}
  {Asymptotic Theory of Quantum Statistical Inference}}}\ (\bibinfo
  {publisher} {World Scientific Pub Co Pte Lt},\ \bibinfo {year} {2005})\ pp.\
  \bibinfo {pages} {43--63}\BibitemShut {NoStop}%
\bibitem [{\citenamefont {Ogawa}\ and\ \citenamefont
  {Nagaoka}(2005)}]{Ogawa2005}%
  \BibitemOpen
  \bibfield  {author} {\bibinfo {author} {\bibfnamefont {T.}~\bibnamefont
  {Ogawa}}\ and\ \bibinfo {author} {\bibfnamefont {H.}~\bibnamefont
  {Nagaoka}},\ }in\ \href {\doibase 10.1142/9789812563071_0003} {\emph
  {\bibinfo {booktitle} {Asymptotic Theory of Quantum Statistical Inference}}}\
  (\bibinfo  {publisher} {World Scientific Pub Co Pte Lt},\ \bibinfo {year}
  {2005})\ pp.\ \bibinfo {pages} {28–--42}\BibitemShut {NoStop}%
\bibitem [{\citenamefont {Brandao}\ and\ \citenamefont
  {Plenio}(2010)}]{GeneralStein}%
  \BibitemOpen
  \bibfield  {author} {\bibinfo {author} {\bibfnamefont {F.~G.}\ \bibnamefont
  {Brandao}}\ and\ \bibinfo {author} {\bibfnamefont {M.~B.}\ \bibnamefont
  {Plenio}},\ }\href@noop {} {\bibfield  {journal} {\bibinfo  {journal}
  {Commun. Math. Phys.}\ }\textbf {\bibinfo {volume} {295}},\ \bibinfo {pages}
  {791} (\bibinfo {year} {2010})}\BibitemShut {NoStop}%
\bibitem [{\citenamefont {Vedral}(2002)}]{Vedral2002}%
  \BibitemOpen
  \bibfield  {author} {\bibinfo {author} {\bibfnamefont {V.}~\bibnamefont
  {Vedral}},\ }\href {\doibase 10.1103/revmodphys.74.197} {\bibfield  {journal}
  {\bibinfo  {journal} {Rev. Mod. Phys.}\ }\textbf {\bibinfo {volume} {74}},\
  \bibinfo {pages} {197–234} (\bibinfo {year} {2002})}\BibitemShut {NoStop}%
\bibitem [{\citenamefont {Vedral}\ \emph {et~al.}(1997)\citenamefont {Vedral},
  \citenamefont {Plenio}, \citenamefont {Rippin},\ and\ \citenamefont
  {Knight}}]{Vedral1997}%
  \BibitemOpen
  \bibfield  {author} {\bibinfo {author} {\bibfnamefont {V.}~\bibnamefont
  {Vedral}}, \bibinfo {author} {\bibfnamefont {M.~B.}\ \bibnamefont {Plenio}},
  \bibinfo {author} {\bibfnamefont {M.~A.}\ \bibnamefont {Rippin}}, \ and\
  \bibinfo {author} {\bibfnamefont {P.~L.}\ \bibnamefont {Knight}},\ }\href
  {\doibase 10.1103/physrevlett.78.2275} {\bibfield  {journal} {\bibinfo
  {journal} {Phys. Rev. Lett.}\ }\textbf {\bibinfo {volume} {78}},\ \bibinfo
  {pages} {2275} (\bibinfo {year} {1997})}\BibitemShut {NoStop}%
\bibitem [{\citenamefont {Brandao}\ \emph {et~al.}(2013)\citenamefont
  {Brandao}, \citenamefont {Horodecki}, \citenamefont {Oppenheim},
  \citenamefont {Renes},\ and\ \citenamefont {Spekkens}}]{Brandao2013}%
  \BibitemOpen
  \bibfield  {author} {\bibinfo {author} {\bibfnamefont {F.~G. S.~L.}\
  \bibnamefont {Brandao}}, \bibinfo {author} {\bibfnamefont {M.}~\bibnamefont
  {Horodecki}}, \bibinfo {author} {\bibfnamefont {J.}~\bibnamefont
  {Oppenheim}}, \bibinfo {author} {\bibfnamefont {J.~M.}\ \bibnamefont
  {Renes}}, \ and\ \bibinfo {author} {\bibfnamefont {R.~W.}\ \bibnamefont
  {Spekkens}},\ }\href@noop {} {\bibfield  {journal} {\bibinfo  {journal}
  {Phys. Rev. Lett.}\ }\textbf {\bibinfo {volume} {111}},\ \bibinfo {pages}
  {250404} (\bibinfo {year} {2013})}\BibitemShut {NoStop}%
\bibitem [{\citenamefont {Brandao}\ and\ \citenamefont
  {Gour}(2015)}]{Brandao2015}%
  \BibitemOpen
  \bibfield  {author} {\bibinfo {author} {\bibfnamefont {F.~G. S.~L.}\
  \bibnamefont {Brandao}}\ and\ \bibinfo {author} {\bibfnamefont
  {G.}~\bibnamefont {Gour}},\ }\href@noop {} {\bibfield  {journal} {\bibinfo
  {journal} {Phys. Rev. Lett.}\ }\textbf {\bibinfo {volume} {115}},\ \bibinfo
  {pages} {070503} (\bibinfo {year} {2015})}\BibitemShut {NoStop}%
\bibitem [{\citenamefont {Junge}\ \emph {et~al.}(2016)\citenamefont {Junge},
  \citenamefont {Renner}, \citenamefont {Sutter}, \citenamefont {Wilde},\ and\
  \citenamefont {Winter}}]{Junge2016}%
  \BibitemOpen
  \bibfield  {author} {\bibinfo {author} {\bibfnamefont {M.}~\bibnamefont
  {Junge}}, \bibinfo {author} {\bibfnamefont {R.}~\bibnamefont {Renner}},
  \bibinfo {author} {\bibfnamefont {D.}~\bibnamefont {Sutter}}, \bibinfo
  {author} {\bibfnamefont {M.~M.}\ \bibnamefont {Wilde}}, \ and\ \bibinfo
  {author} {\bibfnamefont {A.}~\bibnamefont {Winter}},\ }in\ \href {\doibase
  10.1109/isit.2016.7541748} {\emph {\bibinfo {booktitle} {2016 IEEE
  International Symposium on Information Theory (ISIT)}}}\ (\bibinfo
  {publisher} {Institute of Electrical and Electronics Engineers (IEEE)},\
  \bibinfo {year} {2016})\ pp.\ \bibinfo {pages} {2494--2498}\BibitemShut
  {NoStop}%
\bibitem [{\citenamefont {Kastoryano}\ and\ \citenamefont
  {Eisert}(2013)}]{Kastoryano2013}%
  \BibitemOpen
  \bibfield  {author} {\bibinfo {author} {\bibfnamefont {M.~J.}\ \bibnamefont
  {Kastoryano}}\ and\ \bibinfo {author} {\bibfnamefont {J.}~\bibnamefont
  {Eisert}},\ }\href@noop {} {\bibfield  {journal} {\bibinfo  {journal} {J.
  Math. Phys.}\ }\textbf {\bibinfo {volume} {54}},\ \bibinfo {pages} {102201}
  (\bibinfo {year} {2013})}\BibitemShut {NoStop}%
\bibitem [{\citenamefont {Bernigau}\ \emph {et~al.}(2015)\citenamefont
  {Bernigau}, \citenamefont {Kastoryano},\ and\ \citenamefont
  {Eisert}}]{Bernigau2015}%
  \BibitemOpen
  \bibfield  {author} {\bibinfo {author} {\bibfnamefont {H.}~\bibnamefont
  {Bernigau}}, \bibinfo {author} {\bibfnamefont {M.~J.}\ \bibnamefont
  {Kastoryano}}, \ and\ \bibinfo {author} {\bibfnamefont {J.}~\bibnamefont
  {Eisert}},\ }\href {\doibase 10.1088/1742-5468/2015/02/p02008} {\bibfield
  {journal} {\bibinfo  {journal} {J. Stat. Phys.}\ }\textbf {\bibinfo {volume}
  {2015}},\ \bibinfo {pages} {P02008} (\bibinfo {year} {2015})}\BibitemShut
  {NoStop}%
\bibitem [{\citenamefont {Goold}\ \emph {et~al.}(2016)\citenamefont {Goold},
  \citenamefont {Huber}, \citenamefont {Riera}, \citenamefont {del Rio},\ and\
  \citenamefont {Skrzypczyk}}]{Goold2016}%
  \BibitemOpen
  \bibfield  {author} {\bibinfo {author} {\bibfnamefont {J.}~\bibnamefont
  {Goold}}, \bibinfo {author} {\bibfnamefont {M.}~\bibnamefont {Huber}},
  \bibinfo {author} {\bibfnamefont {A.}~\bibnamefont {Riera}}, \bibinfo
  {author} {\bibfnamefont {L.}~\bibnamefont {del Rio}}, \ and\ \bibinfo
  {author} {\bibfnamefont {P.}~\bibnamefont {Skrzypczyk}},\ }\href {\doibase
  10.1088/1751-8113/49/14/143001} {\bibfield  {journal} {\bibinfo  {journal}
  {J. Phys. A}\ }\textbf {\bibinfo {volume} {49}},\ \bibinfo {pages} {143001}
  (\bibinfo {year} {2016})}\BibitemShut {NoStop}%
\bibitem [{\citenamefont {Gallego}\ \emph {et~al.}(2016)\citenamefont
  {Gallego}, \citenamefont {Eisert},\ and\ \citenamefont
  {Wilming}}]{Gallego2015}%
  \BibitemOpen
  \bibfield  {author} {\bibinfo {author} {\bibfnamefont {R.}~\bibnamefont
  {Gallego}}, \bibinfo {author} {\bibfnamefont {J.}~\bibnamefont {Eisert}}, \
  and\ \bibinfo {author} {\bibfnamefont {H.}~\bibnamefont {Wilming}},\ }\href
  {\doibase 10.1088/1367-2630/18/10/103017} {\bibfield  {journal} {\bibinfo
  {journal} {New J. Phys.}\ }\textbf {\bibinfo {volume} {18}},\ \bibinfo
  {pages} {103017} (\bibinfo {year} {2016})}\BibitemShut {NoStop}%
\bibitem [{\citenamefont {Lashkari}\ and\ \citenamefont
  {Van~Raamsdonk}(2016)}]{Lashkari2015}%
  \BibitemOpen
  \bibfield  {author} {\bibinfo {author} {\bibfnamefont {N.}~\bibnamefont
  {Lashkari}}\ and\ \bibinfo {author} {\bibfnamefont {M.}~\bibnamefont
  {Van~Raamsdonk}},\ }\href@noop {} {\bibfield  {journal} {\bibinfo  {journal}
  {J. High En. Phys.}\ }\textbf {\bibinfo {volume} {2016}} (\bibinfo {year}
  {2016})}\BibitemShut {NoStop}%
\bibitem [{\citenamefont {Matsumoto}(2010)}]{Matsumoto2010}%
  \BibitemOpen
  \bibfield  {author} {\bibinfo {author} {\bibfnamefont {K.}~\bibnamefont
  {Matsumoto}},\ }\href@noop {} {\enquote {\bibinfo {title} {Reverse test and
  characterization of quantum relative entropy},}\ } (\bibinfo {year} {2010}),\
  \Eprint {http://arxiv.org/abs/1010.1030} {arXiv:1010.1030} \BibitemShut
  {NoStop}%
\bibitem [{\citenamefont {Audenaert}\ and\ \citenamefont
  {Eisert}(2005)}]{AudenaertContinuous}%
  \BibitemOpen
  \bibfield  {author} {\bibinfo {author} {\bibfnamefont {K.~M.~R.}\
  \bibnamefont {Audenaert}}\ and\ \bibinfo {author} {\bibfnamefont
  {J.}~\bibnamefont {Eisert}},\ }\href@noop {} {\bibfield  {journal} {\bibinfo
  {journal} {J. Math. Phys.}\ }\textbf {\bibinfo {volume} {46}},\ \bibinfo
  {pages} {102104} (\bibinfo {year} {2005})}\BibitemShut {NoStop}%
\bibitem [{\citenamefont {Tomamichel}(2016)}]{Tomamichel2016}%
  \BibitemOpen
  \bibfield  {author} {\bibinfo {author} {\bibfnamefont {M.}~\bibnamefont
  {Tomamichel}},\ }\href {\doibase 10.1007/978-3-319-21891-5} {\emph {\bibinfo
  {title} {Quantum Information Processing with Finite Resources}}},\ \bibinfo
  {series} {SpringerBriefs in Mathematical Physics}, Vol.~\bibinfo {volume}
  {5}\ (\bibinfo  {publisher} {Springer International Publishing},\ \bibinfo
  {year} {2016})\BibitemShut {NoStop}%
\bibitem [{\citenamefont {Janzing}\ \emph {et~al.}(2000)\citenamefont
  {Janzing}, \citenamefont {Wocjan}, \citenamefont {Zeier}, \citenamefont
  {Geiss},\ and\ \citenamefont {Beth}}]{Janzing00}%
  \BibitemOpen
  \bibfield  {author} {\bibinfo {author} {\bibfnamefont {D.}~\bibnamefont
  {Janzing}}, \bibinfo {author} {\bibfnamefont {P.}~\bibnamefont {Wocjan}},
  \bibinfo {author} {\bibfnamefont {R.}~\bibnamefont {Zeier}}, \bibinfo
  {author} {\bibfnamefont {R.}~\bibnamefont {Geiss}}, \ and\ \bibinfo {author}
  {\bibfnamefont {T.}~\bibnamefont {Beth}},\ }\href@noop {} {\bibfield
  {journal} {\bibinfo  {journal} {Int. J. Th. Phys.}\ }\textbf {\bibinfo
  {volume} {39}},\ \bibinfo {pages} {2717} (\bibinfo {year}
  {2000})}\BibitemShut {NoStop}%
\bibitem [{\citenamefont {Horodecki}\ and\ \citenamefont
  {Oppenheim}(2013)}]{Horodecki2013}%
  \BibitemOpen
  \bibfield  {author} {\bibinfo {author} {\bibfnamefont {M.}~\bibnamefont
  {Horodecki}}\ and\ \bibinfo {author} {\bibfnamefont {J.}~\bibnamefont
  {Oppenheim}},\ }\href@noop {} {\bibfield  {journal} {\bibinfo  {journal}
  {Nature Comm.}\ }\textbf {\bibinfo {volume} {4}},\ \bibinfo {pages} {2059}
  (\bibinfo {year} {2013})}\BibitemShut {NoStop}%
\bibitem [{\citenamefont {Brandao}\ \emph {et~al.}(2015)\citenamefont
  {Brandao}, \citenamefont {Horodecki}, \citenamefont {Ng}, \citenamefont
  {Oppenheim},\ and\ \citenamefont {Wehner}}]{Brandao15}%
  \BibitemOpen
  \bibfield  {author} {\bibinfo {author} {\bibfnamefont {F.~G. S.~L.}\
  \bibnamefont {Brandao}}, \bibinfo {author} {\bibfnamefont {M.}~\bibnamefont
  {Horodecki}}, \bibinfo {author} {\bibfnamefont {N.~H.~Y.}\ \bibnamefont
  {Ng}}, \bibinfo {author} {\bibfnamefont {J.}~\bibnamefont {Oppenheim}}, \
  and\ \bibinfo {author} {\bibfnamefont {S.}~\bibnamefont {Wehner}},\
  }\href@noop {} {\bibfield  {journal} {\bibinfo  {journal} {PNAS}\ }\textbf
  {\bibinfo {volume} {112}},\ \bibinfo {pages} {3275} (\bibinfo {year}
  {2015})}\BibitemShut {NoStop}%
\bibitem [{\citenamefont {Gour}(2016)}]{Gour16}%
  \BibitemOpen
  \bibfield  {author} {\bibinfo {author} {\bibfnamefont {G.}~\bibnamefont
  {Gour}},\ }\href@noop {} {\enquote {\bibinfo {title} {Quantum resource
  theories in the single-shot regime},}\ } (\bibinfo {year} {2016}),\ \Eprint
  {http://arxiv.org/abs/1610.04247} {arXiv:1610.04247} \BibitemShut {NoStop}%
\bibitem [{\citenamefont {Buscemi}\ and\ \citenamefont
  {Gour}(2016)}]{Buscemi16}%
  \BibitemOpen
  \bibfield  {author} {\bibinfo {author} {\bibfnamefont {F.}~\bibnamefont
  {Buscemi}}\ and\ \bibinfo {author} {\bibfnamefont {G.}~\bibnamefont {Gour}},\
  }\href@noop {} {\enquote {\bibinfo {title} {{Quantum relative Lorenz
  curves}},}\ } (\bibinfo {year} {2016}),\ \Eprint
  {http://arxiv.org/abs/1607.05735} {arXiv:1607.05735} \BibitemShut {NoStop}%
\bibitem [{\citenamefont {Lostaglio}\ \emph {et~al.}(2015)\citenamefont
  {Lostaglio}, \citenamefont {M\"uller},\ and\ \citenamefont
  {Pastena}}]{Lostaglio15}%
  \BibitemOpen
  \bibfield  {author} {\bibinfo {author} {\bibfnamefont {M.}~\bibnamefont
  {Lostaglio}}, \bibinfo {author} {\bibfnamefont {M.~P.}\ \bibnamefont
  {M\"uller}}, \ and\ \bibinfo {author} {\bibfnamefont {M.}~\bibnamefont
  {Pastena}},\ }\href {\doibase 10.1103/PhysRevLett.115.150402} {\bibfield
  {journal} {\bibinfo  {journal} {Phys. Rev. Lett.}\ }\textbf {\bibinfo
  {volume} {115}},\ \bibinfo {pages} {150402} (\bibinfo {year}
  {2015})}\BibitemShut {NoStop}%
\bibitem [{\citenamefont {Ng}\ \emph {et~al.}(2015)\citenamefont {Ng},
  \citenamefont {Mančinska}, \citenamefont {Cirstoiu}, \citenamefont
  {Eisert},\ and\ \citenamefont {Wehner}}]{Ng15}%
  \BibitemOpen
  \bibfield  {author} {\bibinfo {author} {\bibfnamefont {N.~H.~Y.}\
  \bibnamefont {Ng}}, \bibinfo {author} {\bibfnamefont {L.}~\bibnamefont
  {Mančinska}}, \bibinfo {author} {\bibfnamefont {C.}~\bibnamefont
  {Cirstoiu}}, \bibinfo {author} {\bibfnamefont {J.}~\bibnamefont {Eisert}}, \
  and\ \bibinfo {author} {\bibfnamefont {S.}~\bibnamefont {Wehner}},\ }\href
  {http://stacks.iop.org/1367-2630/17/i=8/a=085004} {\bibfield  {journal}
  {\bibinfo  {journal} {New J. Phys.}\ }\textbf {\bibinfo {volume} {17}},\
  \bibinfo {pages} {085004} (\bibinfo {year} {2015})}\BibitemShut {NoStop}%
\bibitem [{\citenamefont {Sparaciari}\ \emph {et~al.}(2016)\citenamefont
  {Sparaciari}, \citenamefont {Oppenheim},\ and\ \citenamefont
  {Fritz}}]{Sparaciari2016}%
  \BibitemOpen
  \bibfield  {author} {\bibinfo {author} {\bibfnamefont {C.}~\bibnamefont
  {Sparaciari}}, \bibinfo {author} {\bibfnamefont {J.}~\bibnamefont
  {Oppenheim}}, \ and\ \bibinfo {author} {\bibfnamefont {T.}~\bibnamefont
  {Fritz}},\ }\href@noop {} {\enquote {\bibinfo {title} {A resource theory for
  work and heat},}\ } (\bibinfo {year} {2016}),\ \Eprint
  {http://arxiv.org/abs/1607.01302} {1607.01302} \BibitemShut {NoStop}%
\bibitem [{\citenamefont {Sparaciari}\ \emph {et~al.}(2017)\citenamefont
  {Sparaciari}, \citenamefont {Jennings},\ and\ \citenamefont
  {Oppenheim}}]{Sparaciari2017}%
  \BibitemOpen
  \bibfield  {author} {\bibinfo {author} {\bibfnamefont {C.}~\bibnamefont
  {Sparaciari}}, \bibinfo {author} {\bibfnamefont {D.}~\bibnamefont
  {Jennings}}, \ and\ \bibinfo {author} {\bibfnamefont {J.}~\bibnamefont
  {Oppenheim}},\ }\href@noop {} {\enquote {\bibinfo {title} {Energetic
  instability of passive states in thermodynamics},}\ } (\bibinfo {year}
  {2017}),\ \Eprint {http://arxiv.org/abs/1701.01703} {arXiv:1701.01703}
  \BibitemShut {NoStop}%
\end{thebibliography}

%

\newpage
\appendix

\section{Gauge invariance of $M_\beta$}\label{sec:appgauge}
Here, we show that any measure of athermality fulfilling properties \ref{prop:fe:continuity} -- \ref{prop:fe:dpi} is gauge-invariant in the sense that $M_\beta(\rho,H) = M_\beta(\rho,H+C\one)$ for all $c\in\RR$. To see this, first note that since tracing out and adding a thermal ancilla are free transitions, $M_\beta(\omega_{\beta,H},H)=0$ for any $H$.

A simple calculation using additivity then also shows gauge-invariance:
\begin{align}
  M_\beta(\rho,H+C\one) &= M_\beta\left((\rho,H+C\one)\otimes(\omega_{\beta,K},K)\right)\nonumber\\
                        &= M_\beta\left(\rho\otimes \omega_{\beta,K},H\otimes \one + C\one\otimes\one + \one\otimes K\right)\nonumber\\
                        &=M_\beta\left((\rho,H)\otimes(\omega_{\beta,K},K+C\one)\right)\nonumber\\
                        &=M_\beta\left((\rho,H)\otimes(\omega_{\beta,K+C\one},K+C\one)\right)\nonumber\\
                        &= M_\beta(\rho,H),\nonumber
\end{align}
where we made use of the gauge invariance of Gibbs states.

\section{Rank-decreasing quantum channels}
In this appendix we prove the validity of Lemma~\ref{lemma:rank} of the main-text. We have to show that given a channel $T:\mc B(\mc H)\rightarrow \mc B(\mc H')$ and a full-rank state $\sigma$ such that $\mathrm{supp}(T(\sigma))\subset P$, we also have $\mathrm{supp}(T(\rho))\subset P$ for all states $\rho$. Here, $P$ is an arbitrary subspace of the total Hilbert space $\mc H'$.
Let $\sigma = \sum_i q_i \proj{i}$ be the eigen-decomposition of $\sigma$. Since $T$ maps positive operators to positive operators, and the support of the sum of positive operators is the union of the supports of the operators we conclude that $T(\proj{i})$ is supported in $P$ for all $i$. We thus only need to show that also operators of the form $T(\ketbra{i}{j})$ are supported on $P$.
Now consider any density operator $\rho = d+r$ where $d$ is the diagonal part of $\rho$ (in the eigenbasis of $\sigma$) and $r$ the rest. We know that $\tr(T(d)) = 1$ since $T$ is trace-preserving. Hence $\tr(T(r))=0$. Let us now assume (to arrive at a contradiction)
that $T(r)$ has support within the subspace $Q=\one-P$. Since $T$ maps positive operators to positive operators,
\begin{align}
0 \leq Q T(\rho) Q  = Q T(r) Q.
\end{align}
Thus we conclude on the one hand that $Q T(r) Q\geq 0$. On the other hand, we know that
\begin{align}
1=\tr(T(\rho)) \geq \tr(P T(\rho)) = 1 + \tr(PT(r)).
\end{align}
Hence, $\tr(P(T(r)))=0$. Since $T$ is trace-preserving we also have
\begin{align}
\tr(PT(r))= - \tr(Q T(r)) =0.
\end{align}
 Hence $Q T(r)Q=0$ and also $Q T(\rho) Q=0$. By positivity and Hermiticity of $T(\rho)$ we also get $PT(\rho)Q=0$ and $Q T(\rho)P=0$. We thus conclude that $T(\rho)=PT(\rho)P$, which finishes the proof.

\section{Proof of Theorem \ref{thm:equivalencemaintext} and other equivalences}

We will show a more complete set of equivalences than the ones of Theorem \ref{thm:equivalencemaintext}, which corresponds simply to iii)
and iv).

\begin{lemma}[Alternative equivalences]
The following properties are equivalent:
\begin{enumerate}
\item [i)]$M_{\beta}$ fulfills monotonicity \ref{prop:monotonicity} $\iff$ $\mc{M}_{\beta}$ fulfills the data-processing inequality (DPI) \eqref{eq:dataprocessingM}.
\item[ii)] $M_{\beta}$ fulfills catalytic monotonicity \ref{prop:catalyticmonotonicity} and additivity \ref{prop:fe:additivity} $\iff$ $\mc{M}_{\beta}$ fulfills additivity \eqref{eq:additiveM} and the DPI \eqref{eq:dataprocessingM}.
\item[iii)] $M_{\beta}$ fulfills marginal-catalytic monotonicity \ref{prop:mcatalyticmonotonicity} and additivity \ref{prop:fe:additivity} $\iff$ $\mc{M}_{\beta}$ fulfills super-additivity \eqref{eq:superadditiveM}, additivity \eqref{eq:additiveM} and the DPI \eqref{eq:dataprocessingM}.
\item[iv)] $M_{\beta}$ fulfills correlated-catalytic monotonicity \ref{prop:ccatalyticmonotonicity} and additivity \ref{prop:fe:additivity} $\iff$ $\mc{M}_{\beta}$ fulfills super-additivity \eqref{eq:superadditiveM}, additivity \eqref{eq:additiveM} and the DPI \eqref{eq:dataprocessingM}.
\end{enumerate}
\end{lemma}
\begin{proof}
Let us first show i) ($\Rightarrow$). Let $T$ be any given quantum-channel. We have to show that $\mc{M}_\beta(T(\rho),T(\sigma))\leq \mc M_\beta(\rho,\sigma)$. But by the previous discussion, $T(\sigma)$ can always be considered to be full-rank. Therefore, the Hamiltonian $H_{T(\sigma)}$ exists and the map $(\rho, H_{\sigma}) \mapsto (T(\rho),H_{T(\sigma)})$ is automatically a GP-map. We therefore obtain
\begin{align}
\mc{M}_\beta(T(\rho),T(\sigma)) &= M_\beta(T(\rho),H_{T(\sigma)}) \\
&\stackrel{\ref{prop:monotonicity}}{\leq} M_\beta(\rho,H_\sigma) = \mc M_\beta(\rho,\sigma).
\end{align}
i) ($\Leftarrow$) follows as
\begin{eqnarray}
M_{\beta}(\sigma,K)&=&M_{\beta}(G(\rho,\omega_{\beta,H}))\\
&=&\mc{M}_{\beta}(\mc{G}_{\beta}^H(\rho),\omega_{\beta,\bar{\mc{G}}(H)})\\
&\stackrel{\eqref{eq:conditiongpmaps}}{=}& \mc{M}_{\beta}(\mc{G}_{\beta}^H(\rho),\mc{G}_{\beta}^H(\omega_{\beta,H}))\\
&\stackrel{\eqref{eq:dataprocessingM}}{\leq}&  M_{\beta}(\rho,H).
\end{eqnarray}
The proof of ii) ($\Rightarrow$) is trivial given i), since \ref{prop:mcatalyticmonotonicity}$\Rightarrow$\ref{prop:monotonicity} and it follows straightforwardly that \eqref{prop:fe:additivity}$\Rightarrow$\ref{eq:additiveM}. The proof of ii) ($\Leftarrow$) follows from noting that $(\rho,H)\stackrel{c}{>}(\sigma,K)$ implies that there exist $G$ so that
\begin{equation}
G(\rho \otimes \gamma, H+R)=(\sigma,K)\otimes (\gamma,R).
\end{equation}
Hence, we find that
\begin{eqnarray}
\nonumber \mc{M}_{\beta}(\sigma,\omega_{\beta,K})&+&\mc{M}_{\beta}(\gamma,\omega_{\beta,R})\stackrel{\eqref{eq:additiveM}}{=}\mc{M}_{\beta}(\sigma \otimes \gamma,\omega_{\beta,K} \otimes \omega_{\beta,R})\\
&\stackrel{\eqref{eq:dataprocessingM}}{\leq}& \mc{M}_{\beta}(\rho,\omega_{\beta,H})+\mc{M}_{\beta}(\gamma,\omega_{\beta,R}),
\end{eqnarray}
which implies straightforwardly $M_{\beta}(\rho,H)\geq M_{\beta}(\sigma,K)$, that is, \eqref{prop:catalyticmonotonicity}.

Now we show iii) ($\Rightarrow$). Note that \ref{prop:mcatalyticmonotonicity} implies \ref{prop:catalyticmonotonicity}, since a correlated catalyst is a particular case of using a catalyst. Together with the equivalences i) and ii), we should only show super-additivity of Eq. \eqref{eq:superadditiveM}. This follows from the fact that
\begin{align}
	(\rho_{1,2},H_1+H_2)\stackrel{mc}{>}(\rho_1\otimes\rho_2,H_1+H_2).
\end{align}
	To show this, let us choose as catalyst $\gamma=\rho_1 \otimes \rho_2$. The GP map performing the transition is just a swap between the initial system and the catalyst. Hence, the final system is $(\sigma,K)=(\rho_1 \otimes \rho_2,H_1+H_2)$ and the final catalyst $\tilde{\gamma}=(\rho_{1,2},H_1+H_2)$, which clearly fulfills the conditions of Definition \ref{def:marginal-catalyticfreetransition}.
To see iii) ($\Leftarrow$), we first note that \eqref{eq:additiveM} and \eqref{eq:dataprocessingM} already imply \ref{prop:catalyticmonotonicity} and \ref{prop:fe:additivity}. It thus remains to show that adding \eqref{eq:superadditiveM} also implies \ref{prop:mcatalyticmonotonicity}. This follows since  super-additivity of $\mc M_\beta$ together with additivity implies
\begin{align}
M_\beta(\tilde{\gamma}_A,R_A) &= \mc M_\beta (\tilde{\gamma}_A, \otimes_i \omega_{\beta, R^i})\\&\geq \sum_i \mc M_\beta(\gamma^i,R^i) \\&= M_\beta(\otimes_i (\gamma^i,R^i) = \sum_i M_\beta(\gamma^i,R^i).
\end{align}
Finally, let us turn to iv). Again, since correlated catalytic transitions include catalytic transitions, the only non-trivial property left to show here is that \ref{prop:ccatalyticmonotonicity} and \ref{prop:fe:additivity} also imply \eqref{eq:superadditiveM}. To see this consider an initial object $(\rho_{1,2},H_1+H_2)$ together with the catalyst $(\rho_2,H_2)$ and use a similar trick as for marginal catalytic transitions. Since a swap between the second system of the initial object and the catalyst is a Gibbs-preserving transition which leaves the catalyst correlated but otherwise unchanged,  we know that $M_\beta(\rho_{1,2},H_1+H_2) \geq M_\beta(\rho_1\otimes\rho_2,H_1+H_2)$. But then we have
\begin{align}
\mc M_\beta(\rho_{1,2},\omega_{\beta,H_1}\otimes \omega_{\beta,H_2}) &= M_\beta(\rho_{1,2},H_1+H_2) \\
&\geq M_\beta(\rho_1\otimes\rho_2,H_1+H_2)\\& = \mc M_\beta(\rho_1\otimes \rho_2,H_1+H_2),
\end{align}
which completes the argument.
\end{proof}

\clearpage

\end{document}